\documentclass[12pt]{article}%
\usepackage{amssymb}
\usepackage{amsfonts}
\usepackage{amsmath}
\usepackage{graphicx}%
\providecommand{\U}[1]{\protect\rule{.1in}{.1in}}
\newtheorem{theorem}{Theorem}
\newtheorem{acknowledgement}[theorem]{Acknowledgement}

\newtheorem{lemma}[theorem]{Lemma}
\newtheorem{notation}[theorem]{Notation}

\newtheorem{proposition}[theorem]{Proposition}
\newtheorem{remark}[theorem]{Remark}

\newenvironment{proof}[1][Proof]{\noindent\textbf{#1.} }{\ \rule{0.5em}{0.5em}}
\begin{document}

\title{Symplectic Polar Duality, Quantum Blobs, and Generalized Gaussians}
\author{Maurice de Gosson\thanks{Corresponding author: maurice.de.gosson@univie.ac.at}
and Charlyne de Gosson\\University of Vienna\\Faculty of Mathematics (NuHAG)\\Oskar-Morgenstern-Platz 1\\1090 Vienna AUSTRIA}
\maketitle
\tableofcontents

\begin{abstract}
We apply the notion of polar duality from convex geometry to the study of
quantum covariance ellipsoids in symplectic phase space. We consider in
particular the case of \textquotedblleft quantum blobs\textquotedblright%
\ introduced in previous work; quantum blobs are the smallest symplectic
invariant regions of the phase space compatible with the uncertainty principle
in its strong Robertson--Schr\"{o}dinger form. We show that these phase space
units can be characterized by a simple condition of reflexivity using polar
duality, thus improving previous results. We apply these geometric
constructions to the characterization of pure Gaussian states in terms of
partial information on the covariance ellipsoid, which allows us to formulate
statements related to symplectic tomography.

\end{abstract}



\section{Introduction}

In a recent paper \cite{gopolar} we discussed the usefulness of the geometric
notion of polar duality in expressing the uncertainty principle of quantum
mechanics. We suggested that a quantum system localized in the position
representation in a set $X$ cannot be localized in the momentum representation
in a set smaller than its polar dual $X^{\hbar}$, the latter being defined as
the set of all $p$ in momentum space such that $px\leq\hbar$ for all $x\in X$.
In the present work we go several steps further by studying the product sets
$X\times X^{\hbar}$. The first observation is that when $X$ is an ellipsoid,
then the John ellipsoid of $X\times X^{\hbar}$ is a \textquotedblleft quantum
blob\textquotedblright, to which one canonically associates a squeezed
coherent state. This leads us to study more general phase space ellipsoids
$\Omega$ viewed as covariance ellipsoids of a quantum state, and we find that
the usual quantum condition for such ellipsoids can be restated in a simple
way using polar duality between intersections with coordinate planes and
orthogonal projection. Thus, we arrive at a purely geometric characterization
of quantization.

The main results of this paper are:

\begin{itemize}
\item In Theorem \ref{Thm3} we use the notion of \textquotedblleft symplectic
polar duality\textquotedblright\ to characterize those phase space ellipsoids
who arise as covariance ellipsoids $\Omega$ of a quantum state. This result is
very much related to what is called in quantum physics \textquotedblleft
symplectic tomography\textquotedblright\ \cite{ib} since it gives global
information by studying the local information obtained by considering the
intersection of $\Omega$ with a Lagrangian plane;

\item Theorem \ref{Thm1}: we prove that a centered phase space ellipsoid
$\Omega$ is a quantum blob (\textit{i.e.} a symplectic ball with
radius$\sqrt{\hbar}$ \cite{Birk,blob,goluPR}) if and only if the polar dual of
the projection of $\Omega$ on the position space is the intersection of
$\Omega$ with the momentum space; this considerably strengthens a previous
result obtained in \cite{gopolar};

\item Theorem \ref{Thm2}; it is an analytical version of Theorem \ref{Thm1},
which we use to give a simple characterization of pure Gaussian states in
terms of partial information on the covariance ellipsoid of a Gaussian state.
This result is related to the so-called \textquotedblleft Pauli
problem\textquotedblright.
\end{itemize}

\begin{notation}
The configuration space of a system with $n$ degrees of freedom will in
general be written $\mathbb{R}_{x}^{n}$, and its dual (the momentum space)
$\mathbb{R}_{p}^{n}$. The position variables will be written $x=(x_{1}%
,...,x_{n})$ and the momentum variables $p=(p_{1},...,p_{n})$. The duality
form (identified with the usual inner product) is $p\cdot x=p_{1}x_{1}%
+\cdot\cdot\cdot+p_{n}x_{n}$. The product $\mathbb{R}_{x}^{n}\times
\mathbb{R}_{p}^{n}$ is identified with $\mathbb{R}^{2n}$ and is equipped with
the standard symplectic form $\sigma$ defined by $\omega(z,z^{\prime})=p\cdot
x^{\prime}-p^{\prime}\cdot x$ if $z=(x,p)$, $z^{\prime}=(x^{\prime},p^{\prime
})$. The corresponding symplectic group is denoted $\operatorname*{Sp}(n)$:
$S\in\operatorname*{Sp}(n)$ if and only $\omega(Sz,Sz^{\prime})=\omega
(z,z^{\prime})$ for all $z,z^{\prime}$. We denote by $\operatorname*{Sym}%
_{++}(n,\mathbb{R})$ the cone of real positive definite symmetric $n\times n$
matrices, and by $GL(n,\mathbb{R})$ the general (real) linear group (the
invertible real $n\times n$ matrices).
\end{notation}

\section{A Geometric Quantum Phase Space}

\subsection{Polar duality and quantum states}

Let $X\subset\mathbb{R}_{x}^{n}$ be a convex body: $X$ is compact\ and convex
and has non-empty interior $\operatorname*{int}(X)$. If $0\in
\operatorname*{int}(X)$ we define the $\hbar$-polar dual $X^{\hslash}%
\subset\mathbb{R}_{p}^{n}$ of $X$ by
\begin{equation}
X^{\hslash}=\{p\in\mathbb{R}^{m}:\sup\nolimits_{x\in X}(p\cdot x)\leq\hbar\}
\label{omo2}%
\end{equation}
where $\hbar$ is a positive constant (we have $X^{\hslash}=\hbar X^{o}$ where
$X^{o}$ is the traditional polar dual dual from convex geometry). The
following properties of polar duality are obvious \cite{Vershynin}:

\begin{itemize}
\item $(X^{\hslash})^{\hbar}=X$ (reflexivity) and $X\subset Y\Longrightarrow
Y^{\hslash}\subset X^{\hslash}$ (anti-monotonicity),

\item For all $L\in GL(n,\mathbb{R})$:%
\begin{equation}
(LX)^{\hbar}=(L^{T})^{-1}X^{\hslash} \label{scaling}%
\end{equation}
(scaling property). In particular $(\lambda X)^{\hbar}=\lambda^{-1}X^{\hslash
}$ for all $\lambda\in\mathbb{R}$, $\lambda\neq0$.
\end{itemize}

We can view $X$ and $X^{\hslash}$ as subsets of phase space by the
identifications $\mathbb{R}_{x}^{n}\equiv\mathbb{R}_{x}^{n}\times0$ and
$\mathbb{R}_{p}^{n}\equiv0\times\mathbb{R}_{p}^{n}$. Writing $\ell
_{X}=\mathbb{R}_{x}^{n}\times0$ and $\ell_{P}=0\times\mathbb{R}_{p}^{n}$ the
transformation $X\longrightarrow X^{\hslash}$ is a mapping $\ell
_{X}\longrightarrow\ell_{P}$. With this interpretation formula (\ref{scaling})
can be rewritten in symplectic form as%
\begin{equation}
(M_{L^{-1}}X)^{\hbar}=M_{L^{T}}X^{\hslash}\label{ML}%
\end{equation}
where $M_{L}=%
\begin{pmatrix}
L^{-1} & 0\\
0 & L^{T}%
\end{pmatrix}
$ is in $\operatorname*{Sp}(n)$. Notice that $M_{L}:\ell_{X}\longrightarrow
\ell_{X}$ and $M_{L}:\ell_{P}\longrightarrow\ell_{P}$.

Suppose now that $X$ is an ellipsoid centered at the origin:
\begin{equation}
X=\{x\in\mathbb{R}_{x}^{n}:Ax\cdot x\leq\hbar\} \label{Xell}%
\end{equation}
where $A\in\operatorname*{Sym}_{++}(n,\mathbb{R})$. The polar dual
$X^{\hslash}$ is the ellipsoid%
\begin{equation}
X^{\hslash}=\{p\in\mathbb{R}_{p}^{n}:A^{-1}p\cdot p\leq\hbar\}.
\label{Xelldual}%
\end{equation}
In particular the polar dual of the ball $B_{X}^{n}(\sqrt{\hbar}%
)=\{x:|x|\leq\sqrt{\hbar}\}$ is $(B_{X}^{n}(\sqrt{\hbar}))^{\hbar}=B_{P}%
^{n}(\sqrt{\hbar})$.

Let $\Omega$ be a convex body in $\mathbb{R}^{2n}$. Recall \cite{Ball} that
the John ellipsoid $\Omega_{\mathrm{John}}$ is the unique ellipsoid in
$R^{2n}$ with maximum volume contained in $\Omega$. If $M\in GL(2n,\mathbb{R}%
)$ then
\begin{equation}
(M(\Omega))_{\mathrm{John}}=M(\Omega_{\mathrm{John}}). \label{JL1}%
\end{equation}

In previous work \cite{blob,goluPR} we called the image of the phase space
ball $B^{2n}(\sqrt{\hbar})$ by some $S\in\operatorname*{Sp}(n)$ a
\textquotedblleft quantum blob\textquotedblright. Quantum blobs are minimum
quantum uncertainty phase space units, and can be used to restate the
uncertainty principle of quantum mechanics in a symplectically invariant form
\cite{go09}. The product $X\times X^{\hslash}$ contains a \emph{unique}
quantum blob:

\begin{proposition}
\label{Prop1}Let $X=\{x:Ax\cdot x\leq\hbar\}$. The John ellipsoid of the
quantum state $X\times X^{\hslash}$ is a a quantum blob, namely
\begin{equation}
(X\times X^{\hslash})_{\mathrm{John}}=M_{A^{1/2}}(B^{2n}(\sqrt{\hbar}))
\label{Johnma}%
\end{equation}
where $M_{A^{1/2}}=%
\begin{pmatrix}
A^{-1/2} & 0\\
0 & A^{1/2}%
\end{pmatrix}
\in\operatorname*{Sp}(n)$.
\end{proposition}

\begin{proof}
That $S_{A^{1/2}}\in\operatorname*{Sp}(n)$ is clear. Let $B_{X}^{n}%
(\sqrt{\hbar})$ and $B_{P}^{n}(\sqrt{\hbar})$ be the balls with radius
$\sqrt{\hbar}$ in in $\mathbb{R}_{x}^{n}$ and $\mathbb{R}_{p}^{n}$,
respectively. We have, by (\ref{Xell}), (\ref{Xelldual}), and (\ref{JL1}),
\begin{align*}
(X\times X^{\hslash})_{\mathrm{John}} &  =(A^{-1/2}B_{X}^{n}(\sqrt{\hbar
})\times A^{1/2}B_{P}^{n}(\sqrt{\hbar}))_{\mathrm{John}}\\
&  =M_{A^{1/2}}(B_{X}^{n}(\sqrt{\hbar})\times B_{P}^{n}(\sqrt{\hbar
}))_{\mathrm{John}}%
\end{align*}
Let us show that
\[
(B_{X}^{n}(\sqrt{\hbar})\times B_{P}^{n}(\sqrt{\hbar})))_{\mathrm{John}%
}=B^{2n}(\sqrt{\hbar});
\]
this will prove our assertion. The inclusion$\ B^{2n}(\sqrt{\hbar})\subset
B_{X}^{n}(\sqrt{\hbar})\times B_{P}^{n}(\sqrt{\hbar})$ is obvious, and we
cannot have $B^{2n}(R)\subset B_{X}^{n}(\sqrt{\hbar})\times B_{P}^{n}%
(\sqrt{\hbar})$ if $R>1$. Assume now that the John ellipsoid $\Omega
_{\mathrm{John}}$ of $\Omega=B_{X}^{n}(\sqrt{\hbar})\times B_{P}^{n}%
(\sqrt{\hbar})$ is defined by
\[
Ax\cdot x+Bx\cdot p+Cp\cdot p\leq\hbar
\]
where $A,C\in\operatorname*{Sym}_{++}(n,\mathbb{R})$ and $B$ are real $n\times
n$ matrices. Since $\Omega$ is invariant by the transformation
$(x,p)\longmapsto(p,x)$ so is $\Omega_{\mathrm{John}}$ and we must thus have
$A=C$ and $B=B^{T}$. Similarly, $\Omega$ being invariant by the partial
reflection $(x,p)\longmapsto(-x,p)$ we get $B=0$ so $\Omega_{\mathrm{John}}$
is defined by $Ax\cdot x+Ap\cdot p\leq\hbar$. The last step is to observe that
$\Omega$, and hence $\Omega_{\mathrm{John}}$, are invariant under all
symplectic rotations $(x,p)\longmapsto(Hx,HP)$ where $H\in O(n,\mathbb{R})$ so
we must have $AH=HA$ for all $H\in O(n,\mathbb{R})$, but this is only possible
if $A=\lambda I_{n\times n}$ for some $\lambda\in\mathbb{R}$. The John
ellipsoid of $\Omega$ is thus of the type $B^{2n}(\sqrt{\hbar/\lambda})$ for
some $\lambda\geq1$ and this concludes the proof in view of the inclusion
$B^{2n}(\sqrt{\hbar})\subset B_{X}^{n}(\sqrt{\hbar})\times B_{P}^{n}%
(\sqrt{\hbar})$ since we cannot have $\lambda>1$.
\end{proof}

\begin{remark}
\label{Rem1}The John ellipsoid $(X\times X^{\hslash})_{\mathrm{John}}$ is the
set of all $(x,p)\in\mathbb{R}_{z}^{2n}$ such that $Ax\cdot x+A^{-1}p\cdot
p\leq\hbar$. The orthogonal projections of $(X\times X^{\hslash}%
)_{\mathrm{John}}$ on the coordinate planes $\ell_{X}=\mathbb{R}_{x}^{n}%
\times0$ and. $\ell_{P}=0\times\mathbb{R}_{p}^{n}$ are therefore $\Pi
_{X}(X\times X^{\hslash})_{\mathrm{John}}=X$ and $\Pi_{P}(X\times X^{\hslash
})_{\mathrm{John}}=X^{\hslash}$.
\end{remark}

The construction above shows that we have a canonical identification between
the ellipsoids $X=\{x:Ax\cdot x\leq\hbar\}$ and the squeezed coherent states
\begin{equation}
\phi_{A}(x)=(\pi\hbar)^{-n/4}(\det A)^{1/4}e^{-Ax\cdot x/2\hbar}. \label{coh1}%
\end{equation}
In fact, the covariance ellipsoid \cite{Littlejohn,Birk} of $\phi_{A}$ is
precisely the John ellipsoid of the product $X\times X^{\hslash}$ as can be
seen calculating the Wigner transform of $\phi_{A}$
\begin{equation}
W\phi_{A}(z)=(\pi\hbar)^{-n}(\det A)^{1/4}\exp\left[  -\frac{1}{\hbar}(Ax\cdot
x+A^{-1}p\cdot p)\right]  \label{wfa}%
\end{equation}
which corresponds to the canonical bijection%
\[
X\longmapsto(X\times X^{\hslash})_{\mathrm{John}}%
\]
between (centered) configuration space ellipsoids $X$ and John ellipsoids of
$X\times X^{\hslash}$ (we will have more to say about this correspondence in
the forthcoming sections).

\subsection{Symplectic polar duality}

Let $\Omega$ be a symmetric convex body in the phase space $(\mathbb{R}%
^{2n},\omega)$. We define the symplectic polar dual $\Omega^{\hbar,\omega}$ of
$\Omega$ as the set
\begin{equation}
\Omega^{\hbar,\omega}=\{z^{\prime}\in\mathbb{R}^{2n}:\sup\nolimits_{z\in
\Omega}\omega(z,z^{\prime})\leq\hbar\}. \label{omegapol1bis}%
\end{equation}
It is straightforward to verify that $\Omega^{\hbar,\omega}$ is related to the
ordinary polar dual $\Omega^{\hbar}$ (calculated by identifying $\mathbb{R}%
^{2n}$ with its own dual) by the formula
\begin{equation}
\Omega^{\hbar,\omega}=(J\Omega)^{\hbar}=J(\Omega^{\hbar}). \label{omegapol2}%
\end{equation}
The properties of symplectic polar duality are easily deduced from those of
ordinary polar duality. This notion is particularly interesting because it
enjoys a property of \textquotedblleft symplectic covariance\textquotedblright:

\begin{proposition}
Let $S\in\operatorname*{Sp}(n)$ and $\Omega$ a symmetric convex body. (i) We
have
\begin{equation}
(S\Omega)^{\hbar,\omega}=S(\Omega^{\hbar,\omega}) \label{omegapol3}%
\end{equation}
\textit{(ii) }The quantum blobs $S(B^{2n}(\sqrt{\hbar}))$, $S\in
\operatorname*{Sp}(n)$, are the only fixed points of the transformation
$\Omega\longmapsto\Omega^{\hbar,\omega}$.
\end{proposition}

\begin{proof}
(i) The condition $S\in\operatorname*{Sp}(n)$ is equivalent to $S^{T}JS=J$
hence $JS=(S^{T})^{-1}J$. Now, using the scaling property (\ref{scaling}) and
the equality (\ref{omegapol2}) we get
\begin{align*}
(S\Omega)^{\hbar,\omega}  &  =J(S(\Omega))^{\hbar}=J(S^{T})^{-1}(\Omega
^{\hbar})\\
&  =SJ(\Omega^{\hbar})=S(\Omega^{\hbar,\omega})
\end{align*}
which is (\ref{omegapol3}). (ii) In particular, since $B^{2n}(\sqrt{\hbar
})^{\hbar}=B^{2n}(\sqrt{\hbar})$ we have
\begin{equation}
(S(B^{2n}(\sqrt{\hbar})))^{\hbar,\omega}=S(B^{2n}(\sqrt{\hbar})).
\label{fixblob}%
\end{equation}

\end{proof}

Let us introduce some terminology. Let $M\in\operatorname*{Sym}_{++}%
(2n,\mathbb{R})$ and consider the centered phase space ellipsoid ellipsoid
\begin{equation}
\Omega=\{z\in\mathbb{R}^{2n}:Mz\cdot z\leq\hbar\}.
\end{equation}
Setting $M=\frac{1}{2}\hbar\Sigma^{-1}$ we can visualize $\Omega$ as the
covariance matrix of a (classical or quantum) state:%
\begin{equation}
\Omega=\{z\in\mathbb{R}^{2n}:\tfrac{1}{2}\Sigma^{-1}z\cdot z\leq
1\}.\label{covellipse}%
\end{equation}
We will say that $\Omega$ is \textit{quantized} if it contains a quantum blob,
\textit{i.e.} if there exists $S\in\operatorname*{Sp}(n)$ such that
$S(B^{2n}(\sqrt{\hbar}))\subset\Omega$. This condition is equivalent to the
uncertainty principle in its strong Robertson--Schr\"{o}dinger form when
$\Sigma$ is viewed as the covariance matrix of a quantum state
\cite{go09,Birk,goluPR}.

Before we proceed to prove the main results we recall the following symplectic
diagonalization result (\textquotedblleft Williamson
diagonalization\textquotedblright\ \cite{Birk}). For every $M\in
\operatorname*{Sym}_{++}(2n,\mathbb{R})$ there exists $S_{0}\in
\operatorname*{Sp}(n)$ such that%
\begin{equation}
M=S_{0}^{T}DS_{0}\text{ \ , \ }D=%
\begin{pmatrix}
\Lambda^{\omega} & 0_{n\times n}\\
0_{n\times n} & \Lambda^{\omega}%
\end{pmatrix}
\label{Williamson}%
\end{equation}
where $\Lambda^{\omega}=\operatorname*{diag}(\lambda_{1}^{\omega}%
,...,\lambda_{n}^{\omega})$; here $\lambda_{1}^{\omega},...,\lambda
_{n}^{\omega}$ the symplectic eigenvalues of $M$ (\textit{i.e.} the moduli of
the usual eigenvalues of the matrix $JM$; they are the same as those of the
antisymmetric matrix $M^{1/2}JM^{1/2}$ and hence of the type $\pm i\lambda$,
$\lambda>0$).

\begin{proposition}
\label{Prop3}Let $\Omega$ be a non-degenerate phase space ellipsoid. (i)
$\Omega$ is quantized if and only if $\Omega^{\hbar,\omega}\subset\Omega$
(i.e. if and only if $\Omega$ contains a quantum blob $S(B^{2n}(\sqrt{\hbar
}))$). (ii) The equality $\Omega^{\hbar,\omega}=\Omega$ holds if and only if
there exists $S\in\operatorname*{Sp}(n)$ such that $\Omega=S(B^{2n}%
(\sqrt{\hbar}))$ (i.e. if and only if $\Omega$ is a quantum blob).
\end{proposition}

\begin{proof}
(i) Suppose that there exists $S\in\operatorname*{Sp}(n)$ such that
$Q=S(B^{2n}(\sqrt{\hbar}))\subset\Omega$. By the anti-monotonicity of
(symplectic) polar duality this implies that we have $\Omega^{\hbar,\omega
}\subset Q^{\hbar,\omega}=Q\subset\Omega$, which proves the necessity of the
condition. Suppose conversely that we have $\Omega^{\hbar,\omega}\subset
\Omega$. Then
\begin{equation}
\Omega^{\hbar,\omega}=\{z\in\mathbb{R}^{2n}:(-JMJ)z\cdot z\leq\hbar\}
\label{ommom}%
\end{equation}
hence the inclusion $\Omega^{\hbar,\omega}\subset\Omega$ implies that
$M\leq(-JMJ)$ ($\leq$ stands here for the L\"{o}wner ordering). Performing a
symplectic diagonalization (\ref{Williamson}) of $M$ and using the relations
$JS^{-1}=S^{T}J$, $(S^{T})^{-1}J=JS$ this is equivalent to
\[
M=S^{T}DS\leq S^{T}(-JD^{-1}J)S
\]
that is to $D\leq-JD^{-1}J$. In the notation in (\ref{Williamson}) this
implies that we have $\Lambda^{\omega}\leq(\Lambda^{\omega})^{-1}$ and hence
$\lambda_{j}^{\omega}\leq1$ for $1\leq j\leq n$; thus $D\leq I$ and
$M=S^{T}DS\leq S^{T}S$. The inclusion $S(B^{2n}(\sqrt{\hbar}))\subset\Omega$
follows. (ii) The condition is sufficient since $S(B^{2n}(\sqrt{\hbar
}))^{\hbar,\omega}=S(B^{2n}(\sqrt{\hbar}))$. Assume conversely that
$\Omega^{\hbar,\omega}=\Omega$. Then there exists $S\in\operatorname*{Sp}(n)$
such that$Q=S(B^{2n}(\sqrt{\hbar}))\subset\Omega$. It follows that
$\Omega^{\hbar,\omega}\subset Q^{\hbar,\omega}=Q$ hence $\Omega^{\hbar,\omega
}=\Omega\subset Q$ so we must have $\Omega=Q$.
\end{proof}

We are going to prove a stronger statement, which can be seen as a
\textquotedblleft tomographic\textquotedblright\ result since it involves the
intersection of the covariance ellipsoid with a subspace. Recall \cite{Birk}
that a subspace $\ell$ of the symplectic space $(\mathbb{R}^{2n},\omega)$ is a
\textit{Lagrangian plane} if $\dim\ell=n$ and $\omega(z,z^{\prime})=0$ for all
$(z,z^{\prime})\in\ell\times\ell$. The coordinate spaces $\ell_{X}%
=\mathbb{R}_{x}^{n}\times0$ and $\ell_{P}=0\times\mathbb{R}_{p}^{n}$ are
trivially Lagrangian planes. The set of all Lagrangian planes is denoted by
$\operatorname*{Lag}(n)$ and is called the Lagrangian Grassmannian; it can be
equipped with a topology making it diffeomorphic to the homogeneous space
$U(n,\mathbb{C})/O(n,\mathbb{R})$. The symplectic group $\operatorname*{Sp}%
(n)$ acts transitively on $\operatorname*{Lag}(n)$; in particular for every
$S\in\operatorname*{Sp}(n)$ the subspaces $\ell=S\ell_{X}$ and $\ell=S\ell
_{P}$ are Lagrangian planes.

\begin{theorem}
\label{Thm3}(i) The ellipsoid $\Omega$ contains a quantum blob $Q=S(B^{2n}%
(\sqrt{\hbar}))$ ($S\in\operatorname*{Sp}(n)$) if and only if there exists
$\ell\in\operatorname*{Lag}(n)$ such that%
\begin{equation}
\Omega^{\hbar,\omega}\cap\ell\subset\Omega\cap\ell\label{interl}%
\end{equation}
in which case we have $\Omega^{\hbar,\omega}\cap\ell\subset\Omega\cap\ell$ for
all $\ell\in\operatorname*{Lag}(n)$. (ii) The equality $\Omega^{\hbar,\omega
}\cap\ell=\Omega\cap\ell$ holds if and only if $\Omega$ is a quantum blob.
\end{theorem}

\begin{proof}
(i) The necessity of the condition (\ref{interl}) is trivial (Proposition
\ref{Prop3}). Let us prove that the condition is sufficient. Setting
$M=\frac{\hbar}{2}\Sigma^{-1}$ and
\begin{equation}
\Omega_{\Sigma}=\{z:Mz\cdot z\leq\hbar\}=\{z:\frac{1}{2}\Sigma^{-1}z\cdot
z\leq\hbar\} \label{omegas}%
\end{equation}
we have%
\begin{equation}
\Omega_{\Sigma}^{\hbar,\omega}=\{z\in\mathbb{R}^{2n}:(-JM^{-1}J)z\cdot
z\leq\hbar\}. \label{omegasdual}%
\end{equation}
We now perform a symplectic diagonalization (\ref{Williamson}) of $M$, this
leads to%
\begin{equation}
\Omega_{\Sigma}=S^{-1}\Omega_{\hbar D^{-1}/2}\text{ \ },\text{ \ }%
\Omega_{\Sigma}^{\hbar,\omega}=S^{-1}(\Omega_{\hbar D^{-1}/2})^{\hbar,\omega}
\label{oms}%
\end{equation}
where $\Omega_{\hbar D^{-1}/2}$ and its dual are explicitly given by
\begin{align*}
\Omega_{\hbar D^{-1}/2}  &  =\{z\in\mathbb{R}^{2n}:Dz\cdot z\leq\hbar\}\\
(\Omega_{\hbar D^{-1}/2})^{\hbar,\omega}  &  =\{z\in\mathbb{R}^{2n}%
:-JD^{-1}Jz\cdot z\leq\hbar\}\\
&  =\{z\in\mathbb{R}^{2n}:D^{-1}z\cdot z\leq\hbar\}.
\end{align*}
Let us first assume that $\ell=\ell_{X}=\mathbb{R}^{n}\times0$. Then
\[
\Omega_{\hbar D^{-1}/2}\cap\ell_{X}=\{x\in\mathbb{R}^{n}:\Lambda^{\omega
}x\cdot x\leq\hbar\}
\]
and
\[
(\Omega_{\hbar D^{-1}/2})^{\hbar,\omega}\cap\ell_{X}=\{x\in\mathbb{R}%
^{n}:(\Lambda^{\omega})^{-1}x\cdot x\leq\hbar\}.
\]
Now, the condition
\[
(\Omega_{\hbar D^{-1}/2})^{\hbar,\omega}\cap\ell_{X}\subset\Omega_{\hbar
D^{-1}/2}\cap\ell_{X}%
\]
is equivalent to $(\Lambda^{\omega})^{-1}\geq\Lambda^{\omega}$ that is to
$D^{-1}\geq D$, which implies $(\Omega_{\hbar D^{-1}/2})^{\hbar,\omega}%
\subset\Omega_{\hbar D^{-1}/2}$, and $\Omega_{\hbar D^{-1}/2}$ contains a
quantum blob in view of Proposition \ref{Prop3}.. We have thus proven the
theorem in the case where $\Sigma=\hbar D^{-1}/2$ and $\ell=\ell_{X}$. For the
general case we take $\ell=S^{-1}\ell_{X}$ where $S$ is a diagonalizing
matrix; in view of (\ref{oms}) we have
\begin{align*}
\Omega_{\Sigma}\cap\ell &  =S^{-1}\Omega_{\hbar D^{-1}/2}\cap S^{-1}\ell
_{X}=S^{-1}(\Omega_{\hbar D^{-1}/2}\cap\ell_{X})\\
\Omega_{\Sigma}^{\hbar,\omega}\cap\ell &  =S^{-1}(\Omega_{\hbar D^{-1}%
/2})^{\hbar,\omega}\cap S^{-1}\ell_{X}=S^{-1}((\Omega_{\hbar D^{-1}/2}%
)^{\hbar,\omega}\cap\ell_{X})
\end{align*}
and hence $\Omega_{\Sigma}^{\hbar,\omega}\cap\ell\subset\Omega_{\Sigma}%
\cap\ell$ if and only if $(\Omega_{\hbar D^{-1}/2})^{\hbar,\omega}%
\subset\Omega_{\hbar D^{-1}/2}$. It now suffices to apply Proposition
\ref{Prop3}. To prove (ii) it is sufficient to note that the equality
\[
(\Omega_{\hbar D^{-1}/2})^{\hbar,\omega}\cap\ell_{X}=\Omega_{\hbar D^{-1}%
/2}\cap\ell_{X}%
\]
is equivalent to $(\Lambda^{\omega})^{-1}=\Lambda^{\omega}$ that is to
$\Lambda^{\omega}=I_{n\times n}$ since we then have $M=S_{0}^{T}S_{0}$ in view
of (\ref{Williamson}), the proof in the general case is then completed as above.
\end{proof}

\subsection{Polar duality and the symplectic camel}

Symplectic capacities (see for instance \cite{cielibak,goluPR}) are numerical
invariants that serve as a fundamental tool in the study of various symplectic
and Hamiltonian rigidity phenomena; they are closely related to Gromov's
symplectic non-squeezing theorem \cite{gr85}; the latter is often referred to
as the \textquotedblleft principle of the symplectic camel\textquotedblright%
\ \cite{go09,MdGiop,goluPR}.

We denote $\operatorname*{Symp}(n)$ the group of all symplectomorphisms
$(\mathbb{R}_{z}^{2n},\omega)\longrightarrow(\mathbb{R}_{z}^{2n},\omega).$
That is, $f\in\operatorname*{Symp}(n)$ if and only $f$ is a diffeomorphism of
$\mathbb{R}_{z}^{2n}$ whose Jacobian matrix $Df(z)$ is in $\operatorname*{Sp}%
(n)$ for every $z\in\mathbb{R}_{z}^{2n}$.

A (normalized) symplectic capacity on $(\mathbb{R}^{2n},\sigma)$ associates to
every subset $\Omega\subset\mathbb{R}_{z}^{2n}$ a number $c(\Omega
)\in\mathbb{[}0,+\infty\mathbb{]}$ such that the following properties hold:

\begin{description}
\item[SC1] \textit{Monotonicity}: If $\Omega\subset\Omega^{\prime}$ then
$c(\Omega)\leq c(\Omega^{\prime})$;

\item[SC2] \textit{Conformality}: For every $\lambda\in\mathbb{R}$ we have
$c(\lambda\Omega)=\lambda^{2}c(\Omega)$;

\item[SC3] \textit{Symplectic invariance}: $c(f(\Omega))=c(\Omega)$ for every
$f\in\operatorname*{Symp}(n)$;

\item[SC4] \textit{Normalization}: For $1\leq j\leq n$ we have $c(B^{2n}%
(r))=\pi r^{2}=c(Z_{j}^{2n}(r))$ where $Z_{j}^{2n}(r)$ is the cylinder with
radius $r$ based on the $x_{j},p_{j}$ plane.
\end{description}

There exists a symplectic capacity, denoted by $c_{\max}$, such that $c\leq
c_{\max}$ for every symplectic capacity. It is defined by
\begin{equation}
c_{\max}(\Omega)=\inf_{f\in\operatorname*{Symp}(n)}\{\pi r^{2}:f(\Omega
)\subset Z_{j}^{2n}(r)\} \label{cmax}%
\end{equation}
where $Z_{j}^{2n}(r)$ is the phase space cylinder defined by $x_{j}^{2}%
+p_{j}^{2}\leq r^{2}$ and $\operatorname*{Symp}(n)$ the group of all
symplectomorphisms of $\mathbb{R}^{2n}$ equipped with the standard symplectic
structure. Similarly, there exists a smallest symplectic capacity $c_{\min}$,
it is defined by
\[
c_{\min}(\Omega)=\sup_{f\in\operatorname*{Symp}(n)}\{\pi r^{2}:f(B^{2n}%
(r))\subset\Omega\}.
\]
One shows \cite{armios08,arkaos13} that\textit{\ }if $X\subset\mathbb{R}%
_{x}^{n}$ and $P\subset\mathbb{R}_{p}^{n}$ are centrally symmetric convex
bodies then we have
\begin{equation}
c_{\max}(X\times P)=4\hbar\sup\{\lambda>0:\lambda X^{\hbar}\subset P\}.
\label{yaron1}%
\end{equation}
In particular,
\begin{equation}
c_{\max}(X\times X^{\hbar})=4\hbar~. \label{yaron3}%
\end{equation}
One also has the weaker notion of linear symplectic capacity, obtained by
replacing condition (SC3) with

\begin{description}
\item[SC3lin] \textit{Linear} \textit{symplectic invariance}: $c(S(\Omega
))=c(\Omega)$ for every $S\in\operatorname*{Sp}(n)$ and $c(\Omega
+z)=c(\Omega)$ for every $z\in\mathbb{R}^{2n}$.
\end{description}

One then defines the corresponding minimal and maximal linear symplectic
capacities $c_{\min}^{\mathrm{lin}}$ and $c_{\max}^{\mathrm{lin}}$
\begin{align}
c_{\min}^{\mathrm{lin}}(\Omega) &  =\sup_{S\in\operatorname*{Sp}(n)}\{\pi
R^{2}:S(B^{2n}(z,R))\subset\Omega,z\in\mathbb{R}^{2n}\}\label{clin1}\\
c_{\max}^{\mathrm{lin}}(\Omega) &  =\inf_{f\in\operatorname*{Sp}(n)}\{\pi
r^{2}:S(\Omega)\subset Z_{j}^{2n}(z,r),z\in\mathbb{R}^{2n}\}.\label{clin2}%
\end{align}

It turns out that all symplectic capacities agree on ellipsoids. They are
calculated as follows: assume that
\[
\Omega=\{z\in\mathbb{R}^{2n}:Mz\cdot z\leq r^{2}\}
\]
where $M\in\operatorname*{Sym}^{+}(2n,\mathbb{R})$, and let $\lambda
_{1}^{\sigma},\lambda_{2}^{\sigma},...,\lambda_{n}^{\sigma}$ be the symplectic
eigenvalue of $M$, \textit{i.e.} the numbers $\lambda_{j}^{\sigma}>0$ ($1\leq
j\leq n$) such that the $\pm i\lambda_{j}^{\sigma}$ are the eigenvalues of the
antisymmetric matrix $M^{1/2}JM^{1/2}$. Then%
\begin{equation}
c(\Omega)=\pi r^{2}/\lambda_{\max}^{\sigma} \label{capellipse}%
\end{equation}
where $\lambda_{\max}^{\sigma}=\max\{\lambda_{1}^{\sigma},\lambda_{2}^{\sigma
},...,\lambda_{n}^{\sigma}\}$ (see \cite{go09,goluPR}). The following
technical Lemma will allows us to prove a refinement of formula (\ref{yaron3}).

\begin{lemma}
\label{Lemmaclin}Let $\Omega\subset\mathbb{R}^{2n}$ be a centrally symmetric
body. We have
\begin{equation}
c_{\min}^{\mathrm{lin}}(\Omega)=\sup_{S\in\operatorname*{Sp}(n)}\{\pi
R^{2}:S(B^{2n}(R))\subset\Omega\}~. \label{clinmin}%
\end{equation}

\end{lemma}

\begin{proof}
Since $\Omega$ is centrally symmetric we have $S(B^{2n}(z_{0},R))\subset
\Omega$ if and only if $S(B^{2n}(-z_{0},R))\subset\Omega$. The ellipsoid
$S(B^{2n}(R))$ is interpolated between $S(B^{2n}(z_{0},R))$ and $S(B^{2n}%
(-z_{0},R))$ using the mapping $t\longmapsto$ $z(t)=z-2tz_{0}$ where $z\in
S(B^{2n}(z_{0},R))$, and is hence contained in $\Omega$ by convexity.
\end{proof}

\begin{proposition}
Let $c_{\min}^{\mathrm{lin}}$ be the smallest linear symplectic capacity and
$X\subset\mathbb{R}_{x}^{n}$ a centered ellipsoid. We have%
\begin{equation}
c_{\min}^{\mathrm{lin}}(X\times X^{\hbar})=\pi\hbar.\label{clinmax}%
\end{equation}

\end{proposition}

\begin{proof}
In view of Lemma \ref{Lemmaclin} $c_{\min}^{\mathrm{lin}}(X\times X^{\hbar})$
is the greatest number $\pi R^{2}$ such that $X\times X^{\hbar}$ contains a
symplectic ball $S(B^{2n}(R)$, $S\in\operatorname*{Sp}(n)$. In view of
Proposition \ref{Prop1} $M_{A^{1/2}}(B^{2n}(\sqrt{\hbar}))$ is such a
symplectic ball; since it is also the largest ellipsoid contained in $X\times
X^{\hbar}$ we must have
\[
c_{\min}^{\mathrm{lin}}(X\times X^{\hbar})=c_{\min}^{\mathrm{lin}}(M_{A^{1/2}%
}(B^{2n}(\sqrt{\hbar})))=\pi\hbar.
\]

\end{proof}

\section{Projections and Intersections of Quantum Blobs}

In this section we generalize the observation made in Remark \ref{Rem1}.

\subsection{Block matrix notation}

For $M\in\operatorname*{Sym}_{++}(2n,\mathbb{R})$ we consider again the phase
space ellipsoid%
\begin{equation}
\Omega=\{z\in\mathbb{R}^{2n}:Mz\cdot z\leq\hbar\}. \label{Mellipse}%
\end{equation}

Let us write $M=\frac{\hbar}{2}\Sigma^{-1}$ and $\Sigma$ in block-matrix form
\begin{equation}
M=%
\begin{pmatrix}
M_{XX} & M_{XP}\\
M_{PX} & M_{PP}%
\end{pmatrix}
\text{ \ },\text{ \ }\Sigma=%
\begin{pmatrix}
\Sigma_{XX} & \Sigma_{XP}\\
\Sigma_{PX} & \Sigma_{PP}%
\end{pmatrix}
\label{M}%
\end{equation}
where the blocks are $n\times n$ matrices. The condition $M\in
\operatorname*{Sym}_{++}(2n,\mathbb{R})$ ensures us that $M_{XX}>0$,
$M_{PP}>0$, and $M_{PX}=M_{XP}^{T}$ (\textit{resp}. $\Sigma_{XX}>0$,
$\Sigma_{PP}>0$, and $\Sigma_{PX}=\Sigma_{XP}^{T}$; see \cite{zhang}). Using
classical formulas for the inversion of block matrices \cite{Tzon} we have
\begin{equation}
M^{-1}=%
\begin{pmatrix}
(M/M_{PP})^{-1} & -(M/M_{PP})^{-1}M_{XP}M_{PP}^{-1}\\
-M_{PP}^{-1}M_{PX}(M/M_{PP})^{-1} & (M/M_{XX})^{-1}%
\end{pmatrix}
\label{Minverse}%
\end{equation}
where $M/M_{PP}$ and $M/M_{XX}$ are the Schur complements:%
\begin{align}
M/M_{PP} &  =M_{XX}-M_{XP}M_{PP}^{-1}M_{PX}\label{schurm1}\\
M/M_{XX} &  =M_{PP}-M_{PX}M_{XX}^{-1}M_{XP}.\label{schurm2}%
\end{align}
Similarly,%
\begin{equation}
\Sigma^{-1}=%
\begin{pmatrix}
(\Sigma/\Sigma_{PP})^{-1} & -(\Sigma/\Sigma_{PP})^{-1}\Sigma_{XP}\Sigma
_{PP}^{-1}\\
-\Sigma_{PP}^{-1}\Sigma_{PX}(\Sigma/\Sigma_{BB})^{-1} & (\Sigma/\Sigma
_{XX})^{-1}%
\end{pmatrix}
\label{covinv}%
\end{equation}
Notice that these formulas imply%
\begin{gather}
\Sigma_{XX}=\frac{\hbar}{2}(M/M_{PP})^{-1}\text{ },\text{ }\Sigma_{PP}%
=\frac{\hbar}{2}(M/M_{XX})^{-1}\label{msig1}\\
\Sigma_{XP}=-\frac{\hbar}{2}(M/M_{PP})^{-1}M_{XP}M_{PP}^{-1}.\label{msig2}%
\end{gather}
Let $M$ be the symmetric positive definite matrix (\ref{M}). The following
results is well-known (see for instance \cite{gopolar}):

\begin{lemma}
\label{LemmaProj}The orthogonal projections $\Pi_{\ell_{X}}\Omega$ and
$P=\Pi_{\ell_{P}}\Omega$ on the coordinate subspaces $\ell_{X}=\mathbb{R}%
_{x}^{n}\times0$ and $\ell_{P}=0\times\mathbb{R}_{p}^{n}$ of $\Omega$ are the
ellipsoids%
\begin{align}
\Pi_{\ell_{X}}\Omega &  =\{x\in\mathbb{R}_{x}^{n}:(M/M_{PP})x\cdot x\leq
\hbar\}\label{boundb}\\
\Pi_{\ell_{P}}\Omega &  =\{p\in\mathbb{R}_{p}^{n}:(M/M_{XX})p\cdot p\leq
\hbar\}. \label{bounda}%
\end{align}
In terms of the covariance matrix $\Sigma$ and the formulas (\ref{msig1}) this
is%
\begin{align}
\Pi_{\ell_{X}}\Omega &  =\{x\in\mathbb{R}_{x}^{n}:\tfrac{1}{2}\Sigma_{XX}%
^{-1}x\cdot x\leq1\}\label{boundc}\\
\Pi_{\ell_{P}}\Omega &  =\{p\in\mathbb{R}_{p}^{n}:\tfrac{1}{2}\Sigma_{PP}%
^{-1}p\cdot p\leq1\}. \label{boundd}%
\end{align}

\end{lemma}

\subsection{Reconstruction of quantum blobs: discussion}

We have seen in Proposition \ref{Prop1} that if $X\subset\mathbb{R}_{x}^{n}$
is a centered ellipsoid then the John ellipsoid of $X\times X^{\hslash}$ is a
a quantum blob. By construction, the orthogonal projections of this quantum
blob on the position and momentum spaces are precisely $X$ and $X^{\hslash}$,
respectively. In this section we address the following question: for a given
ellipsoid $X$ are there other quantum blobs projecting this way? The key to
the answer lies in the following simple observation:

\begin{lemma}
\label{LemmaBlob}The ellipsoid $\Omega$ is a quantum blob $S(B^{2n}%
(\sqrt{\hbar}))$, $S\in\operatorname*{Sp}(n)$ if and only if the block entries
of $M=(SS^{T})^{-1}$ satisfy%
\begin{equation}
M_{XX}M_{PP}-M_{XP}^{2}=I_{n\times n}\text{ , }M_{PX}M_{PP}=M_{PP}M_{XP}.
\label{RSMatrixM}%
\end{equation}
These relations are in turn equivalent to%
\begin{equation}
\Sigma_{XX}\Sigma_{PP}-\Sigma_{XP}^{2}=\tfrac{1}{4}\hbar^{2}I_{n\times
n}\text{ \textit{and} }\Sigma_{PX}\Sigma_{PP}=\Sigma_{PP}\text{ }\Sigma_{XP}
\label{RSMatrix}%
\end{equation}
where $M=(\hbar/2)\Sigma^{-1}.$
\end{lemma}

\begin{proof}
The ellipsoid $\Omega$ is the set of all $z\in\mathbb{R}^{2n}$ such that
$(SS^{T})^{-1}z\cdot z\leq\hbar$. The positive definite matrix $M=(S^{T}%
S)^{-1}$ is thus symplectic. This condition is equivalent to the matrix
relation $MJM=J$, which is itself equivalent to the conditions (\ref{RSMatrix}%
). In this case the matrix $(2/\hbar)\Sigma$ is also symplectic, whence the
conditions (\ref{RSMatrix}).
\end{proof}

\begin{remark}
The conditions (\ref{RSMatrix}) constitute the matrix form of the saturated
Robertson--Schr\"{o}dinger uncertainty principle \cite{go09,goluPR}.
\end{remark}

Explicitly the ellipsoid $\Omega$ is the set of all $z=(x,p)$ such that
$M/M_{PP}$%
\begin{equation}
M_{XX}x\cdot x+(M_{XP}+M_{XP}^{T})x\cdot p+M_{PP}p\cdot p\leq\hslash
;\label{eq}%
\end{equation}
the necessary and sufficient conditions for $\Omega$ to be a quantum blob are
given by the conditions in (\ref{RSMatrixM}) in the lemma above. Let us now
determine the orthogonal projection $X=\Pi_{\ell_{X}}\Omega$ of the quantum
blob $\Omega$ on the position space $\ell_{X}=\mathbb{R}_{x}^{n}\times0$. By
formula (\ref{boundb}) $X$ is the set of all $x$ such that $(M/M_{PP}%
)x^{2}\leq\hbar$. Using the relations (\ref{RSMatrixM}) we have
\begin{align*}
M/M_{PP} &  =M_{XX}-M_{XP}M_{PP}^{-1}M_{PX}\\
&  =(M_{XX}M_{PP}-M_{XP}^{2})M_{PP}^{-1}\\
&  =M_{PP}^{-1}%
\end{align*}
and hence
\begin{align*}
X &  =\{x\in\mathbb{R}_{x}^{n}:M_{PP}^{-1}x\cdot x\leq\hbar\}\\
X^{\hbar} &  =\{p\in\mathbb{R}_{p}^{n}:M_{PP}p\cdot p\leq\hbar\}
\end{align*}
Similarly, the projection $\Pi_{P}\Omega$ on $\ell_{P}=0\times\mathbb{R}%
_{p}^{n}$ is the momentum space ellipsoid
\[
P=\{p\in\mathbb{R}_{p}^{n}:M_{XX}^{-1}p\cdot p\leq\hbar\}.
\]
We thus have $X^{\hbar}=P$ if and only if $M_{PP}M_{XX}=I_{n\times n}$ which
is possible if and only if $M_{XP}=0$, that is, $\Omega$ must be the John
ellipsoid of $X\times X^{\hslash}$. The latter is thus the \emph{only} quantum
blob projecting orthogonally on $X$ and $X^{\hslash}$. This will be discussed
in a more general setting in Theorem \ref{Thm1} below.

Let us next assume that we know: \textit{(i)} the orthogonal projection
$X=\Pi_{\ell_{X}}\Omega$ of the quantum blob $\Omega$ and \textit{(ii)} the
intersections $\Omega\cap\ell_{X}$ and $\Omega\cap\ell_{P}$ of the quantum
blob with the position and momentum spaces:
\begin{gather}
\Omega\cap\ell_{X}=\{x\in\mathbb{R}_{x}^{n}:M_{XX}x\cdot x\leq\hbar
\}\label{omex}\\
\Omega\cap\ell_{P}=\{p\in\mathbb{R}_{p}^{n}:M_{PP}p\cdot p\leq\hbar\}.
\label{omep}%
\end{gather}
We observe that the knowledge of these intersections is not sufficient to
determine $\Omega$. We have to complement these with the first relation
(\ref{RSMatrixM}) to get the lacking
\'{}%
term $M_{XP}$. The solution is however not unique; for instance in the case
$n=1$ we have two solutions $M_{XP}=\pm(M_{XX}M_{PP}-1)^{1/2}$ and the number
of solutions increases with $n$. Observe that the case $M_{XP}=0$ precisely
corresponds to the John ellipsoid. This is closely related to the Pauli
problem \cite{Pauli} for generalized Gaussians..

\subsection{Intersections with Lagrangian planes}

Orthogonal projections and intersections are exchanged by polar duality:

\begin{proposition}
\label{Propinter}(i) For every linear subspace $F$ of $\mathbb{R}^{n}$ we
have
\begin{equation}
(X\cap\ell)^{\hbar}=\Pi_{\ell}(X^{\hslash})\text{ \textit{and} }(\Pi_{\ell
}X)^{\hbar}=X^{\hslash}\cap\ell\label{projint}%
\end{equation}
where $\Pi_{\ell}$ is the orthogonal projection $\mathbb{R}_{x}^{n}%
\longrightarrow\ell$. (In both equalities, the operation of taking the polar
set in the left hand side is made inside $\ell$). (ii) Let $\ell$ be a linear
subspace of $\mathbb{R}^{2n}$ and $\Omega$ a symmetric convex body in
$\mathbb{R}^{2n}$. We have
\begin{equation}
(\Omega\cap\ell)^{\hbar,\omega}=\Pi_{J\ell}(\Omega^{\hbar,\omega})\text{ and
}(\Pi_{J\ell}\Omega)^{\hbar,\omega}=\Omega^{\hbar,\omega}\cap\ell
\label{projinter}%
\end{equation}
where $JF$ is the orthogonal subspace to $F$.
\end{proposition}

\begin{proof}
(i) (See Vershynin \cite{Vershynin}). Let us first show that $\Pi_{\ell
}(X^{\hslash})\subset(X\cap\ell)^{\hbar}$. Let $p\in X^{\hslash}$. We have,
for every $x\in X\cap\ell$,
\[
x\cdot\Pi_{\ell}p=\Pi_{\ell}x\cdot p=x\cdot p\leq\hbar
\]
hence $\Pi_{\ell}p\in(X\cap\ell)^{\hbar}$. To prove the inverse inclusion we
note that it is sufficient, by the anti-monotonicity property of polar
duality, to prove that $(\Pi_{\ell}(X^{\hslash}))^{\hbar}\subset X\cap\ell$.
Let $x\in(\Pi_{\ell}(X^{\hslash}))^{\hbar}$; we have $x\cdot\Pi_{\ell}%
p\leq\hbar$ for every $p\in X^{\hslash}$. Since $x\in\ell$ (because the dual
of a subset of $\ell$ is in $\ell$) we also have
\[
\hbar\geq x\cdot\Pi_{\ell}p=\Pi_{\ell}x\cdot p=x\cdot p
\]
from which follows that $x\in(X^{\hbar})^{\hbar}=X$, which shows that $x\in
X\cap\ell$. This completes the proof of the first formula in (\ref{projint}).
The second formula in (\ref{projint}) follows by duality, noting that in view
of the reflexivity of polar duality we have%
\[
(X^{\hslash}\cap\ell)^{\hbar}=\Pi_{\ell}(X^{\hslash})^{\hbar}=\Pi_{\ell}X
\]
and hence $X^{\hslash}\cap\ell=(\Pi_{\ell}X)^{\hbar}$. (ii) We have
$(\Omega\cap\ell)^{\hbar}=\Pi_{\ell}(\Omega^{\hbar})$ and hence
\[
(\Omega\cap\ell)^{\hbar,\omega}=J(\Omega\cap\ell)^{\hbar}=J\Pi_{\ell}%
(\Omega^{\hbar})
\]
hence the first formula (\ref{projinter}) noting that%
\[
\Pi_{\ell}(\Omega^{\hbar})=J\Pi_{\ell}J^{-1}(J\Omega^{\hbar})=\Pi_{J\ell
}\Omega^{\hbar,\omega}.
\]
The second formula (\ref{projinter}) follows by duality.
\end{proof}

The following result considerably improves the statements we gave in
\cite{gopolar}:

\begin{theorem}
\label{Thm1}A centered phase space ellipsoid $\Omega=\{z:Mz\cdot z\leq\hbar\}$
($M\in\operatorname*{Sym}_{++}(2n,\mathbb{R})$) is a quantum blob if and only
if the equivalent conditions
\begin{equation}
(\Pi_{\ell_{X}}\Omega)^{\hbar}=\Omega\cap\ell_{P}\text{ \ , \ }(\Pi_{\ell_{X}%
}\Omega)^{\hbar,\omega}=(J\Omega)\cap\ell_{X}. \label{bon1}%
\end{equation}
are satisfied. In terms of the matrix $M$ these conditions are equivalent to
the identity%
\begin{equation}
M_{PP}(M/M_{PP})=I_{n\times n}. \label{cond0}%
\end{equation}

\end{theorem}

\begin{proof}
That both conditions (\ref{bon1}) are equivalent is from definition
(\ref{omegapol2}) of symplectic polar duality. Writing $M$ in block matrix
form, the condition $z=(x,p)\in\Omega$ means that%
\[
M_{XX}x^{2}+(M_{XP}+M_{PX})xp+M_{PP}p^{2}\leq\hbar
\]
(we are using the abbreviations $M_{XX}x\cdot x=M_{XX}x^{2}$, etc.) and the
intersection $\Omega\cap\ell_{P}$ is therefore the set
\[
\Omega\cap\ell_{P}=\{p:M_{PP}p^{2}\leq\hbar\}.
\]
On the other hand, in view of Lemma \ref{LemmaProj},%
\[
\Pi_{\ell_{X}}\Omega=\{x:(M/M_{PP})x^{2}\leq\hbar\}
\]
and the polar dual $(\Pi_{\ell_{X}}\Omega)^{\hbar}$ is
\[
(\Pi_{\ell_{X}}\Omega)^{\hbar}=\{p:(M/M_{PP})^{-1}p^{2}\leq\hbar\}
\]
so we have to prove that $\Omega$ is a quantum blob if and only if
(\ref{cond0}) holds. Using the explicit expression (\ref{schurm1}) of the
Schur complement this is equivalent to the condition%
\begin{equation}
(M_{XX}-M_{XP}M_{PP}^{-1}M_{PX})M_{PP}=I_{n\times n}. \label{cond1}%
\end{equation}
Assume now that $\Omega$ is a quantum blob; then $\Omega=S(B^{2n}(\sqrt{\hbar
}))$ for some $S\in\operatorname*{Sp}(n)$; then $z\in\Omega$ if and only if
$Mz\cdot z\leq\hbar$ where $M=(S^{T})^{-1}S^{-1}$. Since $M\in
\operatorname*{Sp}(n)\cap\operatorname*{Sym}_{++}(2n,\mathbb{R})$ we have
$M_{PP}M_{XP}=M_{PX}M_{PP}$ (second formula (\ref{RSMatrixM}) in Lemma
\ref{LemmaBlob}) and hence
\begin{align*}
(M_{XX}-M_{XP}M_{PP}^{-1}M_{PX})M_{PP}  &  =M_{XX}M_{PP}-M_{XP}M_{PP}%
^{-1}(M_{PX}M_{PP})\\
&  =M_{XX}M_{PP}-(M_{XP})^{2}.
\end{align*}
Using the first formula (\ref{RSMatrixM}) in Lemma \ref{LemmaBlob} we thus
have
\begin{equation}
(M_{XX}-M_{XP}M_{PP}^{-1}M_{PX})M_{PP}=I_{n\times n} \label{mim}%
\end{equation}
which implies that $(\Pi_{\ell_{X}}\Omega)^{\hbar}=\Omega\cap\ell_{P}$, so we
have proven the necessity of the condition (\ref{bon1}). Let us prove that
this condition is sufficient as well. In view of Williamson's diagonalization
result (\ref{Williamson}) we have $M=S_{0}^{T}%
\begin{pmatrix}
\Lambda^{\omega} & 0\\
0 & \Lambda^{\omega}%
\end{pmatrix}
S_{0}$ for some $S_{0}\in\operatorname*{Sp}(n)$ where $\Lambda^{\omega}$ is
the diagonal matrix whose non-zero entries are the symplectic eigenvalues of
$M$.. Since a symplectic automorphism transforms a quantum blob into another
quantum blob, we can reduce the proof of the sufficiency of (\ref{bon1}) to
the case where $\Omega$ is the ellipsoid
\[
\Omega_{0}=\{z\in\mathbb{R}^{2n}:\Lambda^{\omega}x^{2}+\Lambda^{\omega}%
p^{2}\leq\hbar\}.
\]
We have here $\Pi_{\ell_{X}}\Omega_{0}=\{x:\Lambda^{\omega}x^{2}\leq\hbar\}$
hence
\[
(\Pi_{\ell_{X}}\Omega_{0})^{\hbar}=\{x:(\Lambda^{\omega})^{-1}x^{2}\leq
\hbar\}
\]
and $\Omega\cap\ell_{P}=\{x:\Lambda^{\omega}x^{2}\leq\hbar\}$. The equality
$(\Pi_{\ell_{X}}\Omega_{0})^{\hbar}=\Omega\cap\ell_{P}$ thus implies that
$\Lambda^{\omega}=I_{n\times n}$ hence $M=S_{0}^{T}S_{0}\in\operatorname*{Sp}%
(n)$ and $\Omega$ is thus the quantum blob $S_{0}^{-1}(B^{2n}(\sqrt{\hbar}))$.
\end{proof}

\section{Gaussian Quantum Phase Space}

In this section we apply some of our previous geometric results to the theory
of Gaussian states.

\subsection{Generalized Gaussians and their Wigner transforms}

Recall that the Wigner transform (or function) of a square integrable function
$\psi:\mathbb{R}^{n}\longrightarrow\mathbb{C}$ is the function $W\psi\in
C^{0}(\mathbb{R}^{2n},\mathbb{R)}$ defined by the absolutely convergent
integral
\begin{equation}
W\psi(x,p)=\left(  \tfrac{1}{2\pi\hbar}\right)  ^{n}\int e^{-\frac{i}{\hbar
}py}\psi(x+\tfrac{1}{2}y)\psi^{\ast}(x-\tfrac{1}{2}y)d^{n}y~. \label{wigtra}%
\end{equation}
The Wigner transform satisfies the Moyal identity%
\begin{equation}
(W\psi|W\phi)_{L^{2}(\mathbb{R}^{2n})}=(2\pi\hbar)^{-n}|(\psi|\phi
)|_{L^{2}(\mathbb{R}^{n})}^{2} \label{Moyal}%
\end{equation}
which implies, in particular, that
\begin{equation}
||W\psi||_{L^{2}(\mathbb{R}^{2n})}=(2\pi\hbar)^{-n/2}||\psi||_{L^{2}%
(\mathbb{R}^{n})}^{2}. \label{oupsi}%
\end{equation}

An important property satisfied by the Wigner transform is its symplectic
covariance: for every $S\in\operatorname*{Sp}(n)$ and $\psi\in L^{2}%
(\mathbb{R}^{n})$ we have
\begin{equation}
W\psi(S^{-1}z)=W(\widehat{S}\psi)(z) \label{symco}%
\end{equation}
where $\widehat{S}\in\operatorname*{Mp}(n)$ is one of the two metaplectic
operators projecting onto $S$ (recall \cite{Birk} that $\operatorname*{Mp}%
(n)$, the metaplectic group, is a unitary representation in $L^{2}%
(\mathbb{R}^{n})$ of the double cover of $\operatorname*{Sp}(n)$). The
covering projection $\pi^{\operatorname*{Mp}}:\operatorname*{Mp}%
(n)\longrightarrow\operatorname*{Sp}(n)$ is uniquely determined by its action
of the generators of $\operatorname*{Mp}(n)$.

Here is a basic example. Let $X\in\operatorname*{Sym}_{++}(n,\mathbb{R})$ and
$Y\in\operatorname*{Sym}(n,\mathbb{R})$. The associated generalized Gaussian
$\psi_{X,Y}$ is defined by%
\begin{equation}
\psi_{X,Y}(x)=(\pi\hbar)^{-n/4}(\det X)^{1/4}e^{-\tfrac{1}{2\hbar}(X+iY)x^{2}%
}. \label{fay}%
\end{equation}
Its Wigner transform is given by \cite{Bas,Birk,Wigner}
\begin{equation}
W\psi_{X,Y}(z)=(\pi\hbar)^{-n}e^{-\tfrac{1}{\hbar}Gz\cdot z} \label{phagauss}%
\end{equation}
where
\begin{equation}
G=%
\begin{pmatrix}
X+YX^{-1}Y & YX^{-1}\\
X^{-1}Y & X^{-1}%
\end{pmatrix}
. \label{gsym}%
\end{equation}
It is essential to observe that $G=G^{T}\in\operatorname*{Sp}(n)$; this is
most easily seen using the factorization
\begin{equation}
G=S^{T}S\text{ \ },\text{ }S=%
\begin{pmatrix}
X^{1/2} & 0\\
X^{-1/2}Y & X^{-1/2}%
\end{pmatrix}
\in\operatorname*{Sp}(n). \label{bi}%
\end{equation}

It can be shown by a direct calculation that the generalized Gaussians satisfy
the second order partial differential equation $\widehat{H}_{X,Y}\psi_{X,Y}=0$
where $\widehat{H}_{X,Y}$ is the operator with Weyl symbol
\begin{equation}
H_{X,Y}(x,p)=(p+Yx)\cdot(p+Yx)+X^{2}x\cdot x-\hbar\operatorname*{Tr}X.
\label{gf3}%
\end{equation}
($H_{X,Y}$ is the \textquotedblleft Fermi function\textquotedblright%
\ \cite{best} of $\psi_{X,Y}$).

Let us introduce the following notation:

\begin{itemize}
\item $\operatorname*{Gauss}_{0}(n)$ is the set of all centered Gaussian
functions (\ref{fay}): $\psi\in\operatorname*{Gauss}_{0}(n)$ if and only if
there exist $X,Y$ such that $\psi=\psi_{X,Y}$;

\item $\operatorname*{Quant}_{0}(n)$ is the set of all centered quantum blobs:
$Q\in\operatorname*{Quant}(n)$ if and only if there exists $S\in
\operatorname*{Sp}(n)$ such that $Q=S(B^{2n}(\sqrt{\hbar}))$.
\end{itemize}

In \cite{blob} we proved that:

\begin{proposition}
There exists a bijection
\[
F:\operatorname*{Gauss}\nolimits_{0}(n)\longrightarrow\operatorname*{Quant}%
\nolimits_{0}(n).
\]
That bijection is defined as follows: if $\psi_{X,Y}$ satisfies
\[
W\psi_{X,Y}(z)=(\pi\hbar)^{-n}e^{-\tfrac{1}{\hbar}Gz\cdot z}%
\]
then $F[\psi_{X,Y}]=\{z:Gz\cdot z\leq\hbar\}]$.
\end{proposition}

That $F[\psi_{X,Y}]\in\operatorname*{Quant}\nolimits_{0}(n)$ immediately
follows from (\ref{bi}).

\subsection{Gaussian density operators}

Let $\widehat{\rho}\in\mathcal{L}^{1}(L^{2}(\mathbb{R}^{n}))$ be a trace class
operator on $L^{2}(\mathbb{R}^{n})$. If $\operatorname*{Tr}(\widehat{\rho})=1$
and $\widehat{\rho}$ is positive semidefinite ($\widehat{\rho}\geq0$) one says
that $\widehat{\rho}$ is a density operator (it represents the mixed states in
quantum mechanics, for an up-to-date discussion of trace class operators and
their applications to quantum mechanics see \cite{QHA}). One shows, using the
spectral theorem for compact operators, that the Weyl symbol of $\widehat{\rho
}$ can be written as $(2\pi\hbar)^{n}\rho$ where $\rho$ (the \textquotedblleft
Wigner distribution of $\widehat{\rho}$\textquotedblright) is a convex sum%
\[
\rho=\sum_{j}\lambda_{j}W\psi_{j}\text{ , }\lambda_{j}\geq0\text{ , }\sum
_{j}\lambda_{j}=1
\]
where $(\psi_{j})_{j}$ is an orthonormal set of vectors in $L^{2}%
(\mathbb{R}^{n})$ (the series is absolutely convergent in $L^{2}%
(\mathbb{R}^{n})$). Of particular interest are Gaussian density operators, by
definition these are the density operators whose Wigner distribution can be
written
\begin{equation}
\rho(z)=\frac{1}{(2\pi)^{n}\sqrt{\det\Sigma}}e^{-\frac{1}{2}\Sigma
^{-1}(z-z_{0})(z-z_{0})} \label{rhoG}%
\end{equation}
where $z_{0}\in\mathbb{R}_{z}^{2n}$ and the covariance matrix $\Sigma
\in\operatorname*{Sym}_{++}(n,\mathbb{R})$ (we will from now on choose
$z_{0}=0$, but all the statements on the covariance matrix and ellipsoid that
follow are not influenced by this assumption). While the operator
$\widehat{\rho}$ with Weyl symbol $(2\pi\hbar)^{n}\rho$ automatically has
trace one, the condition $\widehat{\rho}\geq0$ is equivalent to
\cite{cogoni,dutta,Birk}
\begin{equation}
\Sigma+\frac{i\hbar}{2}J\geq0. \label{quant0}%
\end{equation}
(that is, the eigenvalues of the Hermitian matrix $\Sigma+\frac{i\hbar}{2}J$
are $\geq0$).

By definition the purity of a density operator $\rho$ is the number
$\mu(\widehat{\rho})=\operatorname*{Tr}(\widehat{\rho}^{2})$. We have
$0<\mu(\widehat{\rho})\leq1$ and $\mu(\widehat{\rho})=1$ if and only the
Wigner distribution $\rho$ of $\widehat{\rho}$ consists of a single term:
$\rho=W\psi$ for some $\psi\in L^{2}(\mathbb{R}^{n})$.

\begin{proposition}
\label{PropBlob}Let $\widehat{\rho}$ be a Gaussian density operator with
covariance matrix $\Sigma$. (i) The condition $\Sigma+\frac{i\hbar}{2}J\geq0$
holds if and only if the covariance ellipsoid $\Omega$ associated with
$\Sigma$ contains a quantum blob. (ii) We have $\mu(\widehat{\rho})=1$ if and
only $\Omega$ is a quantum blob and we have in this case $\rho=W\psi_{X,Y}$
for some pair of matrices $(X,Y)$.
\end{proposition}

\begin{proof}
We have proven part (i) in \cite{Birk,go09} (also see \cite{goluPR}). To prove
(ii) we note that the purity of a Gaussian state $\widehat{\rho}$ is
\cite{Birk}%
\[
\mu(\widehat{\rho})=\left(  \frac{\hbar}{2}\right)  ^{n}(\det\Sigma)^{-1/2}%
\]
hence $\mu(\widehat{\rho})=1$ if and only if $\det\Sigma=(\hbar/2)^{2n}$. Let
$\lambda_{1}^{\omega},...,\lambda_{n}^{\omega}$ be the symplectic eigenvalues
of $\Sigma$ as in the proof of Theorem \ref{Thm1}; in view of Williamson's
symplectic diagonalization theorem there exists $S\in\operatorname*{Sp}(n)$
such that $\Sigma=S^{-1}D(S^{T})^{-1}$ where $D=%
\begin{pmatrix}
\Lambda^{\omega} & 0\\
0 & \Lambda^{\omega}%
\end{pmatrix}
$ with $\Lambda^{\omega}=\operatorname*{diag}(\lambda_{1}^{\omega}%
,...,\lambda_{n}^{\omega})$. The quantum condition (\ref{quant0}) is
equivalent to $\lambda_{j}^{\omega}\geq\hbar/2$ for all $j$ hence
\[
\det\Sigma=(\lambda_{1}^{\sigma})^{2}\cdot\cdot\cdot(\lambda_{n}^{\sigma}%
)^{2}=1
\]
if and only if $\lambda_{j}^{\omega}=\hbar/2$ for all $j$, hence $\Sigma
=\frac{\hbar}{2}S^{-1}(S^{T})^{-1}$ and $\Omega=S(B^{2n}(\sqrt{\hbar}))$ is a
quantum blob.
\end{proof}

\subsection{A characterization of Gaussian density operators}

We are going to apply Theorem \ref{Thm1} to characterize pure Gaussian density
operators without prior knowledge of the full covariance matrix. This is
related to the so-called \textquotedblleft Pauli reconstruction
problem\textquotedblright\ \cite{Pauli} we have discussed in \cite{gopauli}.
The latter can be reformulated in terms of the Wigner transform as follows:
given a function $\psi\in L^{1}(\mathbb{R}^{n})\cap L^{2}(\mathbb{R}^{n})$
whose Fourier transform is also in $L^{1}(\mathbb{R}^{n})\cap L^{2}%
(\mathbb{R}^{n})$ the question is whether we reconstruct $\psi$ from the
knowledge of the marginal distributions%
\begin{equation}
\int W\psi(x,p)d^{n}p=|\psi(x)|^{2}\text{ \ },\text{ \ }\int W\psi
(x,p)d^{n}x=|\widehat{\psi}(p)|^{2} \label{marg}%
\end{equation}
where the Fourier transform $\widehat{\psi}$ of $\psi$ is given by
\begin{equation}
\widehat{\psi}(p)=\left(  \frac{1}{2\pi\hbar}\right)  ^{n/2}\int e^{-\frac
{i}{\hbar}px}\psi(x)d^{n}x. \label{FT}%
\end{equation}
The answer to Pauli's question is negative; the study of this problem has led
to many developments, one of them being the theory of symplectic quantum
tomography (see \textit{e.g}. \cite{ib}). The following result is essentially
an analytic restatement of Theorem \ref{Thm1}:

\begin{theorem}
\label{Thm2}Let $\widehat{\rho}\in\mathcal{L}^{1}(L^{2}(\mathbb{R}^{n}))$ be a
density operator with Gaussian Wigner distribution
\[
\rho(z)=\frac{1}{(2\pi)^{n}\sqrt{\det\Sigma}}e^{-\frac{1}{2}\Sigma^{-1}z\cdot
z}.
\]
Then $\widehat{\rho}$ is a pure density operator if and only if%
\begin{equation}
\Phi(x)=2^{n}\int\rho(x,p)d^{n}p \label{fiw}%
\end{equation}
where $\Phi$ is the Fourier transform of the function $p\longmapsto
\rho(0,p/2)$.
\end{theorem}

\begin{proof}
We begin by noting that by the well-known formula about marginals in
probability theory we have%
\begin{equation}
\int\rho(x,p)d^{n}p=\frac{1}{(2\pi)^{n/2}\sqrt{\det\Sigma_{XX}}}e^{-\tfrac
{1}{2}\Sigma_{XX}^{-1}x\cdot x}. \label{psix}%
\end{equation}
Returning to the notation $M=\frac{\hbar}{2}\Sigma^{-1}$ we have%
\[
\rho(z)=(\pi\hbar)^{-n}(\det M)^{1/2}e^{-\frac{1}{\hbar}Mz\cdot z}%
\]
and the margin formula (\ref{psix}) reads%
\begin{equation}
\int\rho(x,p)d^{n}p=(\pi\hbar)^{-n/2}(\det M/M_{PP})^{1/2}e^{-\frac{1}{\hbar
}(M/M_{PP})x\cdot x}. \label{romm}%
\end{equation}
Assume now that $\widehat{\rho}$ is a pure density operator and let us show
that (\ref{fiw}) holds (also see Remark \ref{Rem2} below). In view of
Proposition \ref{PropBlob} we then have $\rho=W\psi_{X,Y}$ for some Gaussian
(\ref{fay}) and thus $\rho(z)=(\pi\hbar)^{-n}e^{-\tfrac{1}{\hbar}Gz\cdot z}$
where $G$ is the symmetric symplectic matrix (\ref{gsym}). Using the first
marginal property (\ref{marg}) and the definition of $\psi_{X,Y}$ it follows
that
\[
\int\rho(x,p)d^{n}p=|\psi_{X,Y}(x)|^{2}=(\pi\hbar)^{-n/2}(\det X)^{1/2}%
e^{-\tfrac{1}{\hbar}Xx\cdot x}.
\]
On the other hand
\[
W\psi_{X,Y}(0,p/2)=(\pi\hbar)^{-n}e^{-\frac{1}{4\hbar}X^{-1}p\cdot p}%
\]
and its Fourier transform is%
\[
\Phi(x)=\left(  \frac{2}{\pi\hbar}\right)  ^{n}(\det X)^{1/2}e^{-\frac
{1}{\hbar}Xx\cdot x}%
\]
hence the equality (\ref{fiw}). Assume now that, conversely, (\ref{fiw})
holds. We have \ \ \
\[
\rho(0,p/2)=(\pi\hbar)^{-n}(\det M)^{1/2}e^{-\frac{1}{4\hbar}M_{PP}p\cdot p}.
\]
and the Fourier transform $\Phi$ of the function $p\longmapsto\rho(0,p/2)$ is
given by
\[
\Phi(p)=\left(  \frac{2}{\pi\hbar}\right)  ^{n}(\det M)^{1/2}(\det
M_{PP})^{-1/2}e^{-\frac{1}{\hbar}M_{PP}^{-1}x\cdot x}.
\]
The equality (\ref{fiw}) requires that
\[
(\det M)^{1/2}(\det M_{PP})^{-1/2}e^{-\frac{1}{\hbar}M_{PP}^{-1}x\cdot
x}=(\det M/M_{PP})^{1/2}e^{-\frac{1}{\hbar}(M/M_{PP})x\cdot x}%
\]
that is, equivalently,
\begin{gather*}
M_{PP}^{-1}=(M/M_{PP})\\
(\det M)^{1/2}(\det M_{PP})^{-1/2}=(\det M/M_{PP})^{1/2}.
\end{gather*}
The first of these two conditions implies that the covariance ellipsoid
$\Omega$ is a quantum blob (formula (\ref{cond0})) in Theorem \ref{Thm1}); the
second condition is then automatically satisfied since $\det M=1$ \ in this case.
\end{proof}

\begin{remark}
\label{Rem2}Condition (\ref{fiw}) is actually satisfied by \emph{all} even
Wigner transformations (and hence by all pure density operators corresponding
to an even function $\psi$) . suppose indeed that $\rho=W\psi$ for some
suitable even function $\psi\in L^{2}(\mathbb{R}^{n})$. Then
\[
W\psi(0,p/2)=(\pi\hbar)^{-n}\int e^{\frac{i}{\hbar}p\cdot y}|\psi(y)|^{2}%
d^{n}y;
\]
Taking the Fourier transform of both sides and using the first marginal
property (\ref{marg}) yields the identity (\ref{fiw}).
\end{remark}

\section{Perspectives and Comments}

Among all states (classical, or quantum) the Gaussians are those which are
entirely characterized by their covariance matrices. The notion of polar
duality thus appears informally as being a generalization of the uncertainty
principle of quantum mechanics as expressed in terms of variances and
covariances. Polar duality actually is a more general concept than the usual
uncertainty principle, expressed in terms of covariances and variances of
position and momentum variables (and the derived notion of quantum blob). As
was already in the work of Uffink and Hilgevoord \cite{hi,hiuf}, variances and
covariances are satisfactory measures of uncertainties only for Gaussian (or
almost Gaussian) distribution. For more general distributions having
nonvanishing \textquotedblleft tails\textquotedblright\ they can lead to gross
errors and misinterpretation. Another advantage of the notion of polar duality
is that it might precisely be extended to study uncertainties when
non-Gaussianity appears (for an interesting characterization of
non-Gaussianity see \cite{Link}). Instead of considering ellipsoids $X$ in
configuration space $\mathbb{R}_{x}^{n}$ one might want to consider sets $X$
which are only convex. In this case the polar dual $X^{\hbar}$ is still
well-defined and one might envisage, using the machinery of the Minkowski
functional to generalize the results presented here to general non-centrally
symmetric convex bodies in $\mathbb{R}_{x}^{n}$. The difficulty comes from the
fact that we then need to choose the correct center with respect to which the
polar duality is defined since there is no privileged \textquotedblleft
center\textquotedblright\ \cite{arkami}; different choices may lead to polar
duals with very different sizes and volumes. These are difficult questions,
but they may lead to a better understanding of very general uncertainty
principles for the density operators of quantum mechanics.

\begin{acknowledgement}
Maurice de Gosson has been financed by the Grant P 33447 N of the Austrian
Research Foundation FWF.
\end{acknowledgement}


\begin{thebibliography}{99}                                                                                               %


\bibitem {armios08}S. Artstein-Avidan, V. D. Milman, and Y. Ostrover. The
M-ellipsoid, Symplectic Capacities and Volume. \textit{Comment. Math. Helv}.
\textbf{83}(2), 359--369 (2008)

\bibitem {arkaos13}S. Artstein-Avidan, R. Karasev, and Y. Ostrover. From
Symplectic Measurements to the Mahler Conjecture. \textit{Duke Math. J.}
\textbf{163}(11), 2003--2022 (2014)

\bibitem {arkami}S. Artstein, B. Klartag, and V. Milman. The Santal\'{o} point
of a function, and a functional form of the Santal\'{o} inequality.
\textit{Mathematika} \textbf{51}(1-2), 33--48 (2004)

\bibitem {Ball}K. M. Ball. Ellipsoids of maximal volume in convex bodies.
\textit{Geom. Dedicata.}\textbf{ 41}(2), 241--250 (1992)

\bibitem {Bas}M. J. Bastiaans. Wigner distribution function and its
application to first-order optics, \textit{J. Opt. Soc. Am}. \textbf{69}, 1710 (1979)

\bibitem {best}G. Benenti and G. Strini. Quantum mechanics in phase space:
first order comparison between the Wigner and the Fermi function, \textit{Eur.
Phys. J. D }\textbf{57}, 117--121 (2010)

\bibitem {cogoni}E. Cordero, M. de Gosson, and F. Nicola. On the Positivity of
Trace Class Operators, \textit{Advances in Theoretical and Mathematical
Physics} \textbf{23}(8), 2061--2091 (2019)

\bibitem {cielibak}K. Cieliebak, H. Hofer, Latschev, and F. Schlenk.
Quantitative symplectic geometry. arXiv preprint math/0506191 (2005)

\bibitem {dutta}B. Dutta, N. Mukunda, and R. Simon. The real symplectic groups
in quantum mechanics and optics. \textit{Pramana J. of Phys.} \textbf{45}(6),
471--497 (1995)

\bibitem {MdGiop}M. de Gosson. The symplectic camel and phase space
quantization. \textit{J. Phys.A:Math. Gen}. \textbf{34}(47) (2001)

\bibitem {Birk}M. de Gosson. \textit{Symplectic geometry and quantum
mechanics.} Vol. 166. Springer Science \& Business Media, 2006

\bibitem {go09}M. de Gosson. The Symplectic Camel and the Uncertainty
Principle: The Tip of an Iceberg? \textit{Found. Phys}. \textbf{99}, 194 (2009)

\bibitem {blob}M. de Gosson. Quantum blobs. \textit{Found. Phys.} \textbf{43
}(4), 440--457\textbf{\ }(2013)

\bibitem {Wigner}M. de Gosson. \textit{The Wigner Transform}, Advanced
Textbooks in Mathematics, World Scientific, 2017

\bibitem {QHA}M. de Gosson. \textit{Quantum Harmonic Analysis, an
Introduction}, De Gruyter, 2021

\bibitem {gopolar}M. de Gosson. Quantum Polar Duality and the Symplectic
Camel: a New Geometric Approach to Quantization. \textit{Found. Phys}.
\textbf{51}, Article number: 60 (2021)

\bibitem {gopauli}M. de Gosson. The Pauli Problem for Gaussian Quantum States:
Geometric Interpretation. \textit{Mathematics} \textbf{9}(20), 2578 (2021)

\bibitem {goluPR}M. de Gosson and F. Luef. Symplectic Capacities and the
Geometry of Uncertainty: the Irruption of Symplectic Topology in Classical and
Quantum Mechanics. \textit{Phys. Reps.} \textbf{484}, 131--179 (2009)

\bibitem {gr85}M. Gromov. Pseudoholomorphic curves in symplectic manifolds.
\textit{Inv. Math.} \textbf{82}(2), 307--347 (1985)

\bibitem {hi}J. Hilgevoord. The standard deviation is not an adequate measure
of quantum uncertainty. \textit{Am. J. Phys.} \textbf{70}(10), 983 (2002)

\bibitem {hiuf}J. Hilgevoord and J. B. M. Uffink. Uncertainty Principle and
Uncertainty Relations. \textit{Found. Phys}.\textbf{ 15}(9) 925 (1985)

\bibitem {Link}V. Link and W.T. Strunz. Geometry of Gaussian quantum states,
\textit{J. Phys. A: Math. Theor}. \textbf{48} 275301 (2015)

\bibitem {ib}A. Ibort, V. I. Man'ko, G. Marmo, A. Simoni, and F. Ventriglia.
An introduction to the tomographic picture of quantum mechanics, \textit{Phys.
Scr}. \textbf{79}, 065013 (2009)

\bibitem {Littlejohn}R. G. Littlejohn. The semiclassical evolution of wave
packets, \textit{Phys. Reps}. 138(4--5) 193--291 (1986)

\bibitem {Pauli}W. Pauli. General Principles of Quantum Mechanics; Springer
Science \& Business Media: Berlin, Germany, 2012; [Original Title: Prinzipien
der Quantentheorie, Handbuch der Physik, v.5.1, 1958

\bibitem {Tzon}{\normalsize Tzon-Tzer Lu and Sheng-Hua Shiou. Inverses of
$2\times2$ Block Matrices, \textit{Comput. Math. Appl.} \textbf{43}, 119--129
(2002) }

\bibitem {Vershynin}R. Vershynin. Lectures in Geometric Functional Analysis.
Unpublished manuscript. Available at http://www-personal. umich.
edu/romanv/papers/GFA-book/GFA-book. pdf 3.3 (2011)

\bibitem {zhang}F. Zhang. \textit{The Schur Complement and its Applications},
Springer, Berlin, 2005
\end{thebibliography}
\end{document}